\documentclass[a4paper,english]{lipics-v2021}
\hideLIPIcs
\nolinenumbers

\usepackage{microtype}
\usepackage{tikz}
\usetikzlibrary{angles,quotes,calc}

\usepackage{amsfonts}
\usepackage{epsfig}
\usepackage{graphics}
\usepackage{graphicx}
\usepackage{booktabs}
\usepackage{pdfpages}
\usepackage{tikz}
\usepackage{pgfplots}
\usepackage[T1]{fontenc}

\usepackage{complexity}
\usepackage{url}

\usepackage{latexsym}
\usepackage{psfrag}
\usepackage{enumerate}
\usepackage{here}
\usepackage{esvect}
\usepackage[update,prepend]{epstopdf}
\usepackage{anyfontsize}
\usepackage{url}
\usepackage{hyperref}

\usepackage{tikz}
\usepackage{xcolor}

\definecolor{rust}{HTML}{A84B2F}
\definecolor{tealaccent}{HTML}{01696F}
\definecolor{textgray}{HTML}{28251D}
\definecolor{linegray}{HTML}{7A7974}

\newcommand{\ie}{{i.e.}}
\newcommand{\eg}{{e.g.}}

\newcommand{\opt}{\textsf{OPT}}

\newcommand{\conv}{{\rm conv}}

\newcommand{\diam}{{\rm diam}}

\newcommand{\NN}{\mathbb{N}} 
\newcommand{\RR}{\mathbb{R}} 
\newcommand{\eps}{\varepsilon}

\def\I{\mathcal I}
\def\J{\mathcal J}

\newcommand{\rhoC}{\rho_{\rm \, Croft}}

\newcommand{\later}[1]{}
\newcommand{\old}[1]{}

\usepackage[noline,linesnumbered]{algorithm2e}
\SetArgSty{textrm}
\let\oldnl\nl
\newcommand{\nonl}{\renewcommand{\nl}{\let\nl\oldnl}}
\makeatletter
\def\TitleOfAlgo{\@ifnextchar({\@TitleOfAlgoAndComment}{\@TitleOfAlgoNoComment}}
\def\@TitleOfAlgoAndComment(#1)#2{\nonl\hspace*{-1.5em}#2 #1\;}
\def\@TitleOfAlgoNoComment#1{\nonl\hspace*{-1.5em}#1\;}
\makeatother
\DontPrintSemicolon

\title{Avoiding Distances}

\titlerunning{Avoiding Distances}

\author{Adrian Dumitrescu}
{Algoresearch L.L.C., Milwaukee, WI, USA, and 
Research Institute of the University of Bucharest, Romania, and 
Alfr\'ed R\'enyi Institute of Mathematics, Budapest, Hungary}
{ad.dumitrescu@algoresearch.org}
{0000-0002-1118-0321}
{}

\author{Zsolt L\'angi}
       {Bolyai Institute, University of Szeged, Hungary, and
University of Duna\'ujv\'aros, Hungary}
{zlangi@server.math.u-szeged.hu}
{}
{}

\authorrunning{Adrian Dumitrescu and Zsolt L\'angi}

\Copyright{Adrian Dumitrescu and  Zsolt L\'angi}

\funding{Work on this paper was partially supported by the ERC Advanced Grant no. 882971  ``GeoScape",
  and by the Hungarian Science Foundation (NKFIH) grants  131529, 147544, 
and 2024-1.2.8-T\'ET-IPARI-CN-2025-00011.}



\begin{document}

\maketitle


\begin{abstract}
Given a finite planar point set $X=\{x_1,\ldots,x_n\}$, let $a_1<a_2< \cdots < a_s$ denote
the $s \leq {n \choose 2}$ different (Euclidean) distances between points in $X$.
The aim is to retain as many of the points as possible such that certain imposed
inter-pair distances are avoided. 

For example, we show that:

\smallskip
(A) Given a set $X$ of $n$ points in the plane, an independent set of size $0.2293 n$
in the unit distance graph determined by $X$ can be found in $O(n^2)$ time.
For any fixed $k \in \NN$ there exists a constant $c_k>0$
such that for any set $X$ of $n$ points in the plane, and set $D$ of size $k$, there
exists a subset $X' \subseteq X$ of size $|X'| \geq c_k |X|$ in which no pair of points
is at some distance in $D$ (we say that $X'$ \emph{avoids distances} in $D$).
Such a set can be found efficiently.

\smallskip
(B) (i)~There is a subset of $X$ of size $n/2$ that avoids the maximum distance;
this bound cannot be improved.
(ii)~There is a subset of $X$ of size $8n/47$ that avoids both the minimum and maximum distance;
on the other hand, there exists $X$ such that at most $5n/16$ points can be
retained with this property.
(iii)~There is a subset of $X$ of size $n/7$ whose diameter is at most half the original diameter;
on the other hand, there exists $X$ such that at most $n/6$ points can be retained with this property.

\smallskip
(C) (i)~For any $\eps>0$, a $(1-\eps)$-approximation for \textsc{MIS} in the minimum distance graph
of $n$ points in the plane can be found in $2^{O(1/\eps)} n^{O(1)}$ time. 
(ii)~Given a set $X$ of $n$ points in the plane, a maximum independent set in its diameter graph
can be computed in $O(n \log{n})$ time.

\smallskip
The techniques used to handle the planar case can be extended to $\RR^d$ for any fixed dimension $d$.
\end{abstract} 

\maketitle

\section{Introduction} \label{sec:intro}

A subset selection problem is generally of the form: given a finite point
set satisfying condition $\mathcal{A}$, how large of a subset satisfying
condition $\mathcal{B}$ can we find?
For instance if $\mathcal{A}$ is empty and $\mathcal{B}$ is ``no three points are collinear'',
we get the problem of selecting the largest subset in general position, see~\cite{PW13}.
Here we focus on such problems arising from consideration of the set of distances spanned
by an $n$-element point set in the Euclidean space. 

Given a finite planar point set $X=\{x_1,\ldots,x_n\}$, let $A(X) := \{|xy|: x,y\in X,  x\neq y\}$.
The set $A(X)$ is  an  $s$-\emph{distance set} if $|A(X)| = s$, where $s=s(X) \leq {n \choose 2}$. 
In this setting, $A(X)=\{a_1,a_2,\ldots,a_s\}$, where $a_1<a_2< \cdots < a_s$ denote the $s$ different
distances between points in $X$.

We consider several variants. The simplest we know goes back to a question of Erd\H{o}s:
Given $n$ points in $\RR^d$ with minimum distance $1$, let $f(n, d)$ be the largest number such that 
one can retain $f(n, d)$ points so that their minimum distance is $>1$.
It is easy to show that $f(n, d) \leq \lceil \frac{n}{d+1} \rceil$ by considering a set of widely spaced
unit simplices. 
Erd\H{o}s asked for a proof of $f(n, d) \geq c n/d$. Pollack~\cite{Po85} showed that $f(n, 2) \geq n/4$
by a simple planarity argument for the corresponding minimum-distance graph.
No matching bounds are known in any dimension; the current best lower and upper bounds for the planar version, 
$8n/31$ and $5n/16$, are due to  Swanepoel~\cite{Sw02} and Pach and T\'oth~\cite{PT96}, respectively.

An earlier result of Csizmadia~\cite{Cs98} shows the existence of an independent set of size $9n/35$.
Moreover, his proof gives an algorithm for finding a set of this size. His algorithm finds between one and nine
new independent points in each step and runs in $O(n \log{n})$ time. Likewise, a similar algorithm is implicit
in the result of Swanepoel~\cite{Sw02}, as mentioned by the author. 

More generally, given an $n$-element point set $X$, and a subset $D \subseteq A(X)$,
let $G(X,D)$ denote the $n$-vertex graph on $X$ and whose edges are pairs of points in $X$ at a distance
that appears in $D$. 
Further, let $g(X,D)$ denote the maximum size of a subset $Y \subseteq X$ avoiding every distance in $D$;
namely, $g(X,D)$ is the size of a maximum independent set in $G(X,D)$.
Finally, let $g(n,D)$ be the minimum of $g(X,D)$ over all $n$-element point sets $X$;
and $g(n,k)$ be the minimum of $g(X,D)$ over all $n$-element point sets $X$ and $k$-element distance sets
$D \subseteq A(X)$. 

For example, if $X$ is the vertex set of a regular pentagon of unit side-length, then
$g(X,\{1\})=2$ and $g(X,\{1,(\sqrt5 +1)/2\})=1$.
For another example, according to the results on the minimum distance mentioned above, we have
\begin{equation} \label{eq:min}
\frac{8n}{31} \leq g(n,\{a_1\}) \leq \frac{5n}{16}.
\end{equation}

\subsection{Motivating questions and our results}

We consider several questions of the following kind:

\begin{enumerate} \itemsep 2pt
  
\item Given $X$ and $x \in A(X)$, how many points of $X$ can one retain so that among those points
  no pair is at distance $x$? 
  
\item Given $X$ and a set $D$ of $k$ distances, how many points of $X$ can one retain
  so that among those points no pair is at distance $x \in D$?

\item What is the computational complexity of finding a maximum independent set in a distance graph?
  How well can one approximate the maximum independent set in a distance graph?

\end{enumerate}

We give partial answers for some; however, obtaining precise estimates remains elusive for
most of the questions. Throughout this paper we restrict ourselves to $d=2$.
(Higher dimensions are briefly mentioned in Section~\ref{sec:remarks}.)

\medskip
Among others, we show that:

\smallskip
(i) Given a set $X$ of $n$ points in the plane, an independent set of size $0.2293 n$
  in the unit distance graph determined by $X$ can be found in $O(n^2)$ time
  (Theorem~\ref{thm:unit-approx} in Section~\ref{sec:discrete-single}). 

\smallskip
(ii) For any fixed $k \in \NN$ there exists a constant $c_k>0$
such that for any set $X$ of $n$ points in the plane, and set $D$ of size $k$, there
exists a subset $X' \subseteq X$ of size $|X'| \geq c_k |X|$ in which no pair of points
is at some distance in $D$ (we say that $X'$ \emph{avoids distances} in $D$).
One may take $c_k = 0.2293^k$. 
Such a subset can be found efficiently.
(Theorem~\ref{thm:k} in Section~\ref{sec:discrete-multiple}). 

\smallskip
(iii) A $(1/2)$-approximation for \textsc{MIS} in the minimum distance graph of $n$ points in the plane
can be found in $O(n \log{n})$ time. For any $\eps>0$, a $(1-\eps)$-approximation for \textsc{MIS} in
the same graph can be found in $2^{O(1/\eps)} n^{O(1)}$ time.
(Theorem~\ref{thm:min} in Section~\ref{sec:discrete-single}). 

\smallskip
(iv) There is a subset of $X$ of size $n/2$ that avoids the maximum distance; this bound is tight.
  Given a set $X$ of $n$ points in the plane, a maximum independent set in its diameter graph 
  $G(X,\{a_s\})$ can be computed in $O(n \log{n})$ time.
(Theorems~\ref{thm:max} and~\ref{thm:max-exact} in Section~\ref{sec:discrete-single}). 

\smallskip
(v) There is a subset of $X$ of size $8n/47$ that avoids both the minimum and maximum distance;
on the other hand, there exists $X$ such that at most $5n/16$ points can be retained with this property.
(Theorem~\ref{thm:min+max} in Section~\ref{sec:discrete-multiple}).
For any $\eps>0$, a $(0.5 -\eps)$-approximation for \textsc{MIS} in the corresponding distance graph
can be found in $2^{O(1/\eps)} n^{O(1)}$ time.
(Theorem~\ref{thm:min+max-approx} in Section~\ref{sec:discrete-multiple}).

\smallskip
(vi) There is a subset of $X$ of size $n/7$ whose diameter is at most half the original diameter;
on the other hand, there exists $X$ such that at most $n/6$ points can be retained with this property.
(Theorem~\ref{thm:half} in Section~\ref{sec:interval}).

\section{Preliminaries} \label{sec:prelim}

\subsection{Definitions and notations} \label{subsec:def}

In an undirected graph $G=(V,E)$, let $d_v$ denote the degree of
vertex $v$, and let $\alpha(G)$ be the \emph{independence number} of $G$.
The \emph{independence ratio} of $G$ is the ratio $\alpha(G)/|V|$ between the independence number
of $G$ and the order of $G$. 

The following result is due to Caro and Wei~\cite[Chap.~6]{AS16}.
\begin{equation} \label{eq:caro-wei}
\alpha(G) \geq \sum_{v \in V} \frac{1}{d_v+1} . 
\end{equation} 
Let $\overline{d}$ denote the average degree of the corresponding graph. 
Then~\eqref{eq:caro-wei} further yields that
\begin{equation} \label{eq:caro-wei-avg}
\alpha(G) \geq \sum_{v \in V} \frac{1}{d_v+1}  \geq \frac{n}{\overline{d} +1}.
\end{equation} 

The \emph{chromatic number} of $\RR^d$, denoted $\chi (\RR^d)$, is  the minimum number of colors
needed to color the points of $\RR^d$ so that points at unit distance receive different colors; see,
\eg, \cite{Sz02}. The \emph{measurable chromatic number} of $\RR^d$, denoted $\chi_{\rm m}(\RR^d)$,
is the minimum number of colors using Lebesgue measurable color classes.
Clearly $\chi (\RR^d) \leq \chi_{\rm m}(\RR^d)$. 

Our underlying idea for computing a large independent set in a given unit distance graph
relies on (an estimate of) the maximal fraction of $\RR^d$ that can be covered by a measurable set
that avoids distance $1$, denoted $m_1(\RR^d)$. Formally, 
\begin{equation}  \label{eq:m_1}
  m_1(\RR^d) = \sup_{X \subset \RR^d} \limsup_{r \to \infty} \frac{\lambda(X \cap B_d(r))}{\lambda(B_d(r))},
\end{equation}
where $\lambda$ is the $d$-dimensional Lebesgue measure,
$B_d(r)$ is the ball of radius $1$ centered at the origin in $\RR^d$, and
$X$ is a Lebesgue measurable subset of $\RR^d$, such that no two points in $X$ are at distance~$1$.

Clearly $m_1(\RR^d) \in [0,1]$; moreover, it turns out that  $m_1(\RR^d) >0$ for every $d$.
Erd\H{o}s conjectured that $m_1(\RR^d)  < 2^{-d}$ for every dimension $d$~\cite{Sz02}. 
After a long sequence of successive improvements, the planar case was confirmed only recently~\cite{ACM+24}:
$m_1(\RR^2) \leq 0.247$; and it remains open for every $d \geq 3$. 
It follows from definitions that  $\chi_{\rm m}(\RR^d) \geq 1/m_1(\RR^d)$. 
For example, packing disks of diameter  $1$ into the vertices of a triangular lattice with edge length $2$
yields $m_1(\RR^2) \geq \pi/(8 \sqrt3)= 0.2267\ldots$. 
A slightly better construction packs \emph{tortoises} of diameter  $1$ into the vertices of a triangular lattice
with edge length slightly less than $2$. A tortoise is the union of a regular hexagon of diameter $1$ and
a concentric disc of diameter slightly less than $1$.
The resulting density is
\begin{equation} \label{eq:croft}
  \rho_{\rm \,Croft}= \sqrt3 \tan{(\theta/2)} = 0.2293\ldots, \text{ where } \theta + \sin \theta  = \pi/6, \quad 0<\theta< \pi/6.
\end{equation}
This construction, described by Croft~\cite{Cr67}, yields $m_1(\RR^2) \geq \rho_{\rm \,Croft} > 0.2293$.
However, Croft attributes it to Leo Moser; see also~\cite{Sz02}. 

An $\alpha$-approximation algorithm for an optimization problem is a polynomial-time algorithm
that for all instances of the problem produces a solution whose value is within a factor of $\alpha$
of the value of an optimal solution. As in~\cite{WS11}, here we follow the convention that
the approximation ratio of an algorithm for a maximization problem is less than $1$.
Throughout this paper all logarithms are in base~$2$.

\subsection{Related work}

As mentioned in several writings~\cite{CKR18,Ra08,Ra12}
distance graphs are among the most studied objects of combinatorial geometry. First of all,
they are at the ground of the classical Hadwiger–Nelson problem of determining the chromatic number
of the space~\cite{BMP05,CFG91,Ha44,KW91}.

\section{Discrete singleton restrictions} \label{sec:discrete-single}

The idea and quality of approximation of our algorithm is intimately related to implementing the
periodic constructions discussed in Section~\ref{sec:prelim} (with disks or tortoises)
placed at the vertex set of a triangular lattice. 

For any distance $x \in A(X)$, one can find a linear-size subset avoiding that distance
(by Theorem~\ref{thm:unit}). 
From the other direction, there exist $X$, $|X|=n$,  and $x \in A(X)$ such that
$g(X,\{x\}) \leq 2n/7= 0.285\ldots n$. A classic example is as follows: consider $n/7$ disjoint copies
of the Moser spindle with seven points. It is easy to check that each copy can contribute at most two points
to an independent set in $X$. Let $v$ denote the rotation vertex (of degree $4$) of the spindle $S$ and
let $\I \subset S$ be an independent set in the graph.  
If $v \in \I$, only one of the two vertices at distance two from $v$ can be in $\I$.
If $v \notin \I$, $S \setminus \{v\}$ make two disjoint triangles in $S$. At most one vertex
from each can belong to $\I$, as claimed.

About 40 years ago, Erd\H{o}s~\cite{Er87}  asked whether $g(X,\{x\}) \geq n/4$ for all $n$. 
The question received a negative answer, but only very recently~\cite{DV26}. 
This answer is also part of our next result: 

\begin{theorem} \label{thm:unit}
  For any distance $x \in A(X)$ we have $g(n,\{x\}) \geq \rhoC \cdot n \geq 0.2293 n$. 
  On the other hand, there exist $X$, $|X|=n$, and $x \in A(X)$, and $c<1/4$, 
  such that $g(X,\{x\}) \leq c n$.
\end{theorem}
\begin{proof}
\emph{Lower bound.}  After a suitable scaling we may assume that $x=1$.
Let $X$ be any $n$-element point set. 
Consider the periodic tortoise-construction described in the preliminaries.
Let $P$ be a fundamental parallelogram of the lattice generating this arrangement.
Since $X$ is bounded, it is contained in the union of finitely many translates of $P$.
Consider a randomly shifted copy of $X$ by a vector from the origin to a uniformly chosen point of $P$.
The expected number of points in the tortoises is equal to $n$ times the area ratio, namely at least
$0.2293 n$. The lower bound is implied.

\begin{figure}[htbp]
 \centering
 \includegraphics[scale=0.8]{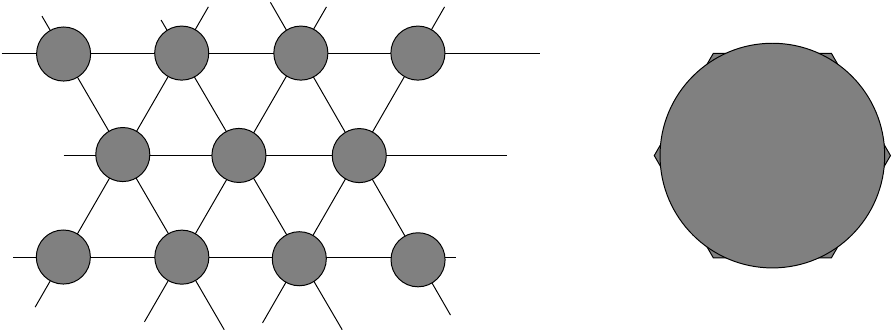}
 \caption{Left: open disks of unit diameter at the vertices of a triangular lattice.
Right: a tortoise.}
 \label{fig:disks}
\end{figure}

\medskip
\emph{Upper bound.} By a very recent result of D{\'u}cz and Varga~\cite{DV26}, there exists a
unit distance graph $G_0$ in the plane such that $\rho_0:= \alpha(G_0)/|V(G_0)| <1/4$.
Let $n_0=|V(G_0)|$. 
Construct an $n$-element point set as follows; for simplicity assume that $n$ is divisible  by $n_0$. 
Consider $n/n_0$ disjoint copies of $G_0$. Each copy can contribute at most $\rho_0 n_0$
points to an independent set in $X$. Overall the size of any independent set is at most
$n/n_0 \cdot \rho_0 n_0 = \rho_0 n $. This proves the upper bound.
Set $c:= \rho_0$. Therefore there exist $X$, $|X|=n$, and $x \in A(X)$, and $c<1/4$, 
  such that $g(X,\{x\}) \leq c n$.
\end{proof}

Until recently, the previous bound implied by a result of Matolcsi, Ruzsa, Varga, and Zs{\'a}mboki~\cite{MRVZ25},
was that for every $\eps>0$, there exist $X$, $|X|=n$, and $x \in A(X)$ such that $g(X,\{x\}) \leq (0.25 +\eps) n$.
The challenging question remaining is that of finding a tight bound.
Before the recent breakthrough of D{\'u}cz and Varga~\cite{DV26} it was less clear what 
the correct bound should be. Perhaps now that the $1/4$ barrier has been crossed,
it can be conjectured that the lower bound implied by Croft's construction should be tight. 

In the next subsection we show that, given a finite point set in the plane,
a $0.2293$-factor approximation of \textsc{MIS} in the corresponding unit distance graph
can be computed efficiently.

\medskip
The case of the minimum distance was mentioned in the introduction. Next we address
the case of the maximum distance.

\begin{theorem} \label{thm:max}
  For the maximum distance $a_s$ we have
  \begin{equation} \label{eq:max}
      g(n,\{a_s \} ) = \left\lfloor \frac{n}{2} \right\rfloor.
  \end{equation}
\end{theorem}
\begin{proof}
  \emph{Lower bound.}
  We use the following result of Dol'nikov~\cite[Corollary 1]{D00}: for any finite point set $X$ in the plane,
  there is a point $x \in X$ such that $X \setminus \{ x \}$ can be partitioned into two parts of smaller diameters.
  Let us choose $x$ according to this, where $X \setminus \{ x \}$ is partitioned into the smaller
  diameter sets $X_1, X_2$. Let the number of elements of $X, X_1, X_2$ be $n, n_1, n_2$, respectively, where
  we assume that $n_1 \leq n_2$. We choose $X_2$ for an independent set. Then, if $n$ is odd, $n_1+n_2 = n-1$ yields
  $n_2 \geq \frac{n-1}{2}$. On the other hand, if $n$ is even, then $n_1 + n_2 =n-1$ is odd, implying that
  $n_1 \neq n_2$; that is, $n_1 \leq n_2-1$. Thus, $2n_2-1 \geq n-1$, from which $n_2 \geq \frac{n}{2}$ follows. 

\medskip
\emph{Upper bound.} Let $X$ be the vertex set of a regular $n$-gon.
 If $n$ is even, the diameter graph consists of a perfect matching of $n/2$ pairwise crossing edges.
 In this case, the size of any independent set is at most $n/2$. 
   If $n$ is odd, the diameter graph consists of a Hamiltonian cycle of $n$ edges, with no two disjoint edges.
 In this case, the size of any independent set is at most $(n-1)/2$. 
\end{proof}

\subsection{Approximation algorithms} \label{subsec:approx}

Our main result in this subsection is the following.

\begin{theorem} \label{thm:unit-approx}
  Given $X$ and $x \in A(X)$, where $|X|=n$, an independent set of size $0.2293 n$ in $G(X,\{x\})$
  can be found in $O(n^2)$ time.
\end{theorem}

We first prove a slightly weaker bound, where $0.2293 n$ is replaced by  $0.2267 n$.
Let $|X|=n$, and let $I =I(X) \subset X$ be an (unknown) maximum independent set in $G(X,\{1\})$.
Recall that packing disks of diameter  $1$ into the vertices of a triangular lattice with edge length $2$
yields $m_1(\RR^2) \geq \pi/(8 \sqrt3)= 0.22672\ldots$. With foresight, we slightly adjust down
the disk size; specifically, set $\eps_1=10^{-5}$, and pack disks of radius $1/2 -\eps_1$;
the resulting lower bound becomes $\pi (1/2 -\eps_1)^2/ (2 \sqrt{3}) =  0.22671\ldots$.
Write $a=0.2267$. By the argument in Theorem~\ref{thm:unit}, we have $\opt=|I| \geq a n$,
and next we give an algorithm that computes an independent set of size $\geq an$ using the slightly
smaller disks. 

\begin{proof}[Proof of Theorem~\ref{thm:unit-approx}]
(i) Consider the following algorithm:

\medskip     
\noindent{{\bf Algorithm~\texttt{A}}} 
Input: a set $X$ of $n$ points in the plane
\begin{itemize} \itemsep 2pt
\item  [] \emph{Step 0}. Set $\eps_1=10^{-5}$ and $r_1 =0.5 -\eps_1$. 
  Lay out (in an arbitrary fashion) a triangular lattice of side length $2$.
  Select a reference (rhombus shaped) cell anchored at the origin and its eight neighbor cells
  in the lattice. Refer to this nine-cell section as the ``board''. 
\item  [] \emph{Step 1.} Map the $n$ input points to the reference cell and replicate the points in this cell
  over the other eight cells of the board. [Overall, the board contains at most $9n$ distinct points.]
\item  [] \emph{Step 2.} Find a (closed) disk of radius $r_1$, call it $D_1$, containing the largest number
  of points on the board (details below). 
\item  [] \emph{Step 3.} Use its center as the translation vector for the lattice:
  translating the lattice by this vector ensures that the disks of radius $r_1$ centered
  at the lattice points collectively contain the maximum number of points from $X$.
\item  [] \emph{Step 4.} Output the subset of $X$ corresponding to the points contained in $D_1$
  as an independent set.  
\end{itemize}

\begin{figure}[htbp]
 \centering
 \includegraphics[scale=0.6]{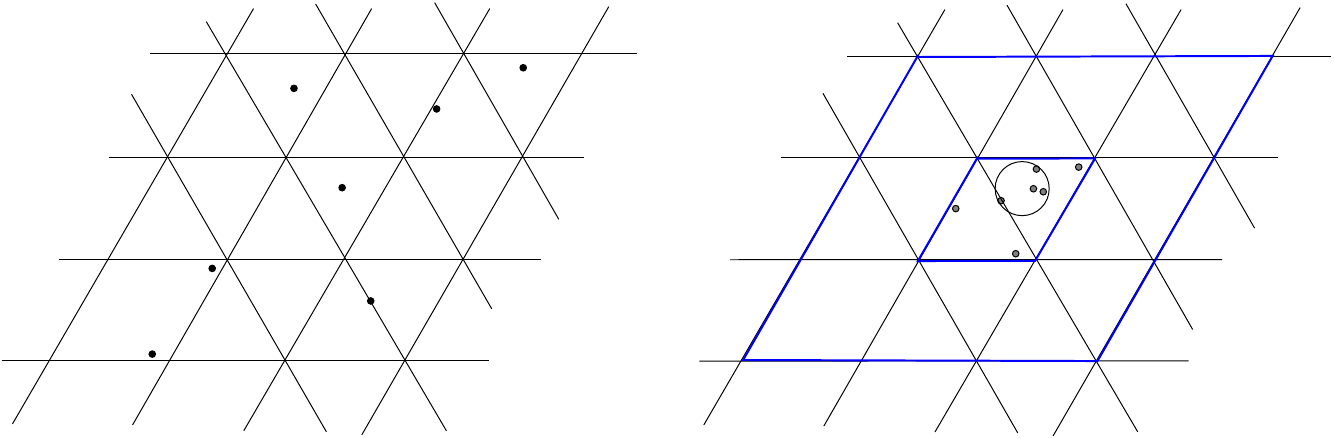}
 \caption{Left: A set of seven points in the plane and the underlying triangular lattice. 
   Right: a~disk of radius $r_1$ enclosing the maximum number of points in the reference cell
   (or its neighbors).} 
 \label{fig:tiling}
\end{figure}

\emph{Algorithm analysis.}
The algorithm correctness follows from the construction: the optimal translate of the system of disks
of radius $r_1$ centered at the vertices of a triangular lattice corresponds to finding a disk of radius $r_1$
that covers the largest number of points on the board. It is worth noting that there is always an optimal disk
that intersects the central cell of the board anchored at the origin. We show that such a disk contains at
least $an$ points. Indeed, the fact that this disk contains $m$ points is equivalent to the fact that a
translate of the original lattice packing of disks containing the largest number of points contains $m$
points. On the other hand, since a random translate of this packing contains any given point with
probability $0.22671\ldots > a$, and by the linearity of expectation, the expected number of points contained
in a random translate of the lattice packing is at least $an$, proving our statement.

Step 2 runs an exact quadratic-time algorithm by Chazelle and Lee~\cite{CL86}.
Gajentaan and Overmars~\cite{GO95}  demonstrated that determining the deepest point
(\ie, covered by the largest number of disks) in an arrangement of $n$ unit disks is 3SUM-hard,
meaning that any exact subquadratic algorithm is highly unlikely.
See~\cite{BCP09} but also~\cite{AW26} for a recent refutation of this conjecture.

The slightly better approximation with ratio $0.2293$ stated in Theorem~\ref{thm:unit-approx}
can be obtained by packing tortoises of unit diameter into the vertices of a triangular lattice
with edge length $2$. Step~2 of Algorithm~\texttt{A} is replaced by:
Find a translate of a tortoise that contains the largest number of points on the board.
The algorithm correctness follows in light of the previous argument.
Finding such a translate amounts to solving a so-called \emph{optimal-placement} problem for maximizing
point containment, for which the algorithm in~\cite[Section~2]{AHR+02} can be used.
Since the boundary of a tortoise consists of $18=O(1)$ edges (namely $12$ line segments and $6$ circular arcs),
the algorithm runs in $O(n^2)$ time~\cite[Corollary~1]{AHR+02}.
\end{proof}

\medskip
\emph{Remark.}
Given a set $X$ of $n$ points, it is natural to ask whether computing a random sample of size $t$
(say, for a small $t \leq n^{1/3}$) may provide a better approximation, if not for finding a \textsc{MIS},
at least for estimating the independence number of $G(X,\{1\})$. For instance, consider an algorithm that
first computes a uniform random sample of a size $t$ and then runs an exact algorithm for \textsc{MIS}
on this sample. One option for the latter step is using the recursive algorithm of Robson~\cite{Ro86},
which runs in $O(2^{0.276 t})$ time on an input of size $t$. Let now $X$ consist of $n/7$ separated copies
of the Moser spindle. A~random uniform sample will select $t$ points from different copies
with high probability, and so the sample will be an independent set with high probability.
Indeed, the probability that the sample contains a pair of points from the same copy is at most
$n/7 \cdot  {7 \choose 2} {n-2 \choose t-2} \big{/} {n \choose t} = O(t^2/n) = o(1)$.
This could suggest the existence of a large independent set in~$X$. However, this intuition is erroneous,
since we have $g(X,\{1\})=2n/7$, which is not far from the minimum possible. A similar outcome can be
enforced for other unit distance graphs.

\subsection{Special cases} \label{subsec:special}

Better approximations or an exact solution can be obtained in special cases, such as
for the minimum or maximum distance, respectively.

\begin{theorem} \label{thm:min}
Given a set $X$ of $n$ points in the plane:
 \begin{enumerate} [{\rm (i)}]
 \item A $1/2$ approximation for \textsc{MIS} in the minimum distance graph $G(X,\{a_1\})$ can be found
   in $O(n \log{n})$ time.
 \item For any $\eps>0$, a $(1-\eps)$-approximation for \textsc{MIS} in $G(X,\{a_1\})$ can be found
   in $2^{O(1/\eps)} n^{O(1)}$ time.
\end{enumerate}
\end{theorem}
\begin{proof}
The graph of minimum distances $G(X,\{a_1\})$ can be computed in $O(n \log{n})$ time by (an extension of)
the classic algorithm of Shamos~\cite[Chap.~5.4]{PS85}.
This graph is plane and so any approximation algorithm for finding a \textsc{MIS} in planar graphs can be used.

(i)~The algorithm of Chiba, Nishizeki, and Saito~\cite{CNS82} computes a $1/2$ approximation for \textsc{MIS}
on any planar graph with $n$ vertices and runs in $O(n \log{n})$ time. 

(ii)~A polynomial time approximation scheme (PTAS) due to Baker~\cite{Ba94} computes a $(1-\eps)$-approximation
in $2^{O(1/\eps)} n^{O(1)}$ time. Note that Baker's algorithm requires a planar embedding of the graph.
As shown by Marx~\cite{Ma07}, this running time is essentially the best possible under the Exponential
Time Hypothesis (ETH).  See also~\cite{DL26} for a recent refutation of this hypothesis.
\end{proof}

\begin{theorem} \label{thm:max-exact}
  Given a set $X$ of $n$ points in the plane, a maximum independent set in
 its diameter graph $G(X,\{a_s\})$ can be computed in $O(n\log n)$ time.
\end{theorem}
\begin{proof}
Let $G=G(X,\{a_s\})$. 
It is well-known that $G$ has at most $n$ edges, no two disjoint edges,
and at most one odd cycle~\cite{Bo90,HP1934,KM01,Pa13,Su1935}.
It is also well-known that $G$ can be computed from $\conv(X)$ in linear time
by rotating calipers thus in $O(n \log{n})$ time overall~\cite[Chap.~4.2]{PS85}. 

Now if $G$ is a forest, then a \textsc{MIS} in $G$ can be computed in $O(n)$ time
by dynamic programming~\cite[Chap.~3.10]{Er19};
so we may assume otherwise. Let $C=\{v_1, v_2, \ldots, v_k\}$ be an (odd) cycle in $G$;
note that $C$ can be computed via a graph traversal~\cite[Chap.~22]{CLRS09}. 
Then $G$ consists of a \emph{hairy cycle}, namely $C$,
with zero or more edges hanging from the cycle, each adjacent to a leaf of degree one.

Removing $v_i$ and its adjacent vertices from $G$ yields a forest $F_i$.
A \textsc{MIS} in $G$ may contain zero or more vertices from $C$.
Its size is $n-k$ in the former case, and  $1 + \max_{i \in [k]} \alpha(F_i)$ in the latter case
(by dynamic programming in the forest $F_i$).
It follows that a \textsc{MIS} in $G$, whose size is
$\max \left( n-k,1 + \max_{i \in [k]} \alpha(F_i) \right)$ can be computed in $O(n^2)$ time.

The running time can be improved as follows:
In each rooted attached star, compute the standard two dynamic-programming
values $I(v)$ and $E(v)$, according as the root $v$ is included or excluded.
\old{
\[
 I(v)=1+\sum_{u\text{ child of }v}E(u),\qquad
 E(v)=\sum_{u\text{ child of }v}\max\{I(u),E(u)\}.
 \]
 } 
For the unique cycle $v_1\ldots v_k$, these computations give
weights $I(v_i),E(v_i)$ for \emph{including} or \emph{excluding} each cycle vertex.
The remaining optimization is the maximum-weight independent set of
that cycle.  Fix $v_1$ as excluded or included (\ie, break the cycle), and solve the resulting
simpler path problem by a two-state dynamic program; take the better result.
All tree and cycle computations take $O(n)$ time.  
Taking into account the time to construct $G$, it follows that a \textsc{MIS} in the diameter graph of $X$
can be computed in $O(n\log n)$ time.
\end{proof}

If the convex hull vertices are provided in cyclic order, the running
time is $O(n)$ in the standard real-RAM model.

\section{Multiple discrete restrictions} \label{sec:discrete-multiple}

\begin{observation} \label{obs:monotone}
The function $g$ satisfies the following properties:
\begin{enumerate} [{\rm (i)}]
\item If $D \subset D'$ then $g(n,D) \geq g(n,D')$.
\item Let $D \subseteq A(X)$. If $x \in D$ and $g(n,x) \geq c n$ for some constant $c>0$,
  then $g(n,D) \geq g(c n, D \setminus \{x\})$. 
  \end{enumerate}
\end{observation}
\begin{proof}
(i) Distances in $D' \setminus D$ may impose additional constraints. 

(ii) Note that removing points does not introduce new distances.
Thus in order to avoid distances in $D$, one can proceed in two steps: first retain $c n$ points
that avoid distance $x \in D$. Then retain (recursively) a subset of these points that avoids distances
in $D \setminus \{x\}$. 
\end{proof}

\begin{theorem} \label{thm:k}
For every $n$ and $k \leq s$:
\begin{enumerate} [{\rm (i)}] \itemsep 2pt
\item If $|D|=k$, then $g(n,D) \geq \rhoC^k n  \geq 0.2293^k n$.  Consequently, $g(n,k) \geq 0.2293^k n$.
  \item Given $X$ and $D \subseteq A(X)$, where $|X|=n$ and $|D|=k$, 
  an independent set of size $0.2293^k n$ can be found in $O(n^2)$ time.
\item There exists $X$ and $D \subseteq  A(X)$ with $|D|=k$,
  such that  $g(X, D) \leq c_k n$,  and $c_k = O \left( (k \sqrt{\log{k}})^{-1} \right)$.
\end{enumerate}
\end{theorem}
\begin{proof}
\emph{Lower bound.}  (i)~Put $a=\rhoC$. Let $D=\{b_1,\ldots,b_k\}$.  Successively avoid the given distances. 
 In the $i$th step, use a randomly shifted periodic tiling of the plane that achieves $m_1(\RR^2) \geq a$.
 Recall that this construction packs \emph{tortoises} of diameter  $1$ into the vertices of a triangular lattice
 with edge length slightly less than $2$. 
Suppose that there are $m$ points at the beginning of the current step.
An averaging argument shows that the expected number of points contained in the tortoises is at least $a m$. 
By design, these points form an independent set.
This completes the description of one step.
The procedure is repeated $k$ times. A fraction $\geq a$ is retained in each step, hence
the lower bound is implied by Observation~\ref{obs:monotone}~(ii). 

\medskip
\emph{Algorithm.} (ii)~Put $a=0.2293$.
Successively avoid the given distances.  In the $i$th step, use the quadratic-time
algorithm in~\cite[Section~2]{AHR+02} from the proof of Theorem~\ref{thm:unit-approx}.
Assuming there are $m$ points at the beginning of the current step,
the algorithm retains  at least $a m$ for the next step in $O(m^2)$ time.
The procedure is repeated $k$ times hence the overall running time is $O(kn^2)$.

The following modification of the algorithm allows for a reduction in its time complexity.
Retaining  \emph{exactly} $m_i=\lceil a m_{i-1} \rceil$ points in each step
preserves avoidance of all previously treated distances and moreover,
we have $m_i\ge a^i n$ by induction.
In addition,  $\lceil a m \rceil \leq m/3$ for every $m \geq 2$, and so
the overall running time is reduced to 
\[ O \left( \sum _{i=0}^\infty 3^{-i} n^2 \right) = O(n^2). \] 

Once a singleton remains, all outstanding distance restrictions are
automatically satisfied, so the algorithm terminates. The inductively
proved lower bound is preserved, completing the proof.

\medskip
\emph{Upper bound.} (iii)~Let $h(k)$ be the maximum size of a planar point set that determines at most $k$ distances.
According to a result of  Erd\H{o}s~\cite{Er46}, a section of the integer lattice gives
$h(k) = \Omega(k \sqrt{\log{k}})$; see also a recent paper by Wang~\cite{Wa25}.
Let $H$ be such a set and set $D=A(H)$.  For simplicity assume that $n$ is divisible  by $|H|$.
Let $X$ consist of $n/|H|$ disjoint copies of $H$. At most one point can be retained from each copy in
a subset that avoids $D$.  This proves the upper bound.
\end{proof}

\subsection{Distance combinations} \label{subsec:combinations}

Consider the min-max distance pair $a_1,a_s$. A two-step elimination procedure using 
Theorem~\ref{thm:max} for the maximum distance, and Observation~\ref{obs:monotone} and~\eqref{eq:min}
for the minimum distance yields $g(n,\{a_1,a_s\}) \geq 4n/31 -O(1)$. Next we obtain a better bound
(observe that the bounds are within a factor of $2$ from each other):

\begin{theorem} \label{thm:min+max}
 For the min-max distance pair $a_1,a_s$ we have
  \begin{equation} \label{eq:min+max}
    \frac{8n-20}{47} \leq g(n,\{a_1,a_s\})  \leq \frac{5n}{16}.
\end{equation}
\end{theorem}
\begin{proof}
  \emph{Lower bound.}
Let $ab$ be an arbitrary diameter pair and let $h$ be the number of vertices of $\conv(X)$.
We may assume that $ab$ is horizontal and that there are $h_1$ vertices on the upper hull and
$h_2$ vertices on the lower hull, where  $h_1 \leq h_2$ and $h_1+ h_2 =h+2$. 
Remove $h_1 \leq h/2 +1$ vertices from the upper hull (including $a$ and $b$), and thereby obtain
an independent set $X_1 \subset X$ in $G(X,\{a_s\})$; where $|X_1| \geq n- h/2 -1$. 

It remains to choose an independent set in $G(X_1,\{a_1\})$. We choose the better of the following two:
\begin{enumerate} \itemsep 3pt
\item Select every other vertex on the lower hull of $X_1$.  
\item Find an independent set of size at least $8|X_1|/31$ in $G(X_1,\{a_1\})$ using the algorithm of
  Swanepoel~\cite{Sw02} mentioned in Section~\ref{sec:intro}. 
\end{enumerate}

Then the size of the resulting subset is at least
\begin{align*}
\max \left\{ \frac{h-3}{4}, \frac{8}{31} \left( n -\frac{h}{2} - 2\right) \right\} \geq 
\frac{16}{47} \cdot \frac{h-3}{4}  + \frac{31}{47} \cdot \frac{8}{31} \cdot \left( n -\frac{h}{2} - 2\right) 
= \frac{8n-20}{47}. 
\end{align*}

\medskip
\emph{Upper bound.}  By Observation~\ref{obs:monotone} and~\eqref{eq:min},
 we have $g(n,\{a_1,a_s\})  \leq g(n,\{a_1\}) \leq \frac{5n}{16}$.
\end{proof}

Next we show that if we forbid the diameter and another distance,  we may be forced
to retain a smaller fraction of the points.

\begin{theorem} \label{thm:any+max}
  For any $x \in A(X)$ let $D=\{x,a_s\}$. Then $g(n,D) \geq \rhoC \cdot n /2 \geq 0.1146 n$.
On the other hand, there exists $X$, $|X|=n$, and $x \in A(X)$, such that  $g(X,\{x,a_s\}) \leq n/4$. 
\end{theorem}
\begin{proof}
  \emph{Lower bound.} We may assume that $1=x <a_s$. Avoid distances in $D$ in two steps.
  Use the procedure from the proof of Theorem~\ref{thm:unit} to avoid $x$ and then keep half
  of the remaining points to avoid the diameter. The procedure retains at least
  $\rhoC \cdot n /2 \geq 0.1146 n$ points.

\medskip
\emph{Upper bound.}
Let $n$ be a multiple of $8$.
Consider the $60^\circ$ rhombus of unit side-length and let $X$ consist of $n/4$ disjoint copies whose
$60^\circ$ vertices form a regular $n/4$-gon . Note that $X$ has $n/4$ diameter pairs, each of length
$a_s=\sqrt{3}$. Let $x=1$ and $D=\{x,a_s\}$. At most one point can be retained from each copy in a subset that
avoids $D$.  This proves the upper bound.
\end{proof}

\begin{theorem} \label{thm:two}
  For any two distances $x,y \in A(X)$ we have $g(n,\{x,y\}) \geq \rhoC^2 \cdot n \geq 0.052 n$.
  On the other hand, there exists $X$ and $x,y \in A(X)$ such that $g(X,\{x,y\}) \leq \lceil n/5 \rceil$. 
\end{theorem}
\begin{proof}
  \emph{Lower bound.}  Use the algorithm from the proof of Theorem~\ref{thm:unit} to avoid
  $x$ and $y$ successively. The procedure retains at least $\rhoC^2 \cdot n \geq 0.052 n$ points.

\medskip
\emph{Upper bound.}  For simplicity assume that $n$ is divisible  by $5$.
Consider $n/5$ disjoint copies of a regular pentagon. Note that it is a two distance set;
let $x,y$ be the two distances. At most one point can be retained from each pentagon in a subset that
avoids $x,y$.  This proves the upper bound.
\end{proof}

In specific instances, we manage to obtain tighter bounds. Theorems~\ref{thm:smallest2} and~\ref{thm:largest2}
are two other cases where the upper and lower bounds are within a factor of $2$ from each other.

\begin{theorem} \label{thm:smallest2}
  For the two smallest distances $a_1,a_2 \in A(X)$ we have $g(n,\{a_1,a_2\}) \geq n/9$.
  On the other hand, there exists $X$, $|X|=n$, such that $g(X,\{a_1,a_2\}) \leq \lceil n/5 \rceil$. 
\end{theorem}
\begin{proof}
\emph{Lower bound.} We may assume that $a_1=1$. 
We distinguish two cases:   (i)~$a_2 \leq \sqrt3$. Since $\chi(\RR^2; [1,\sqrt3]) \leq 9$
by a result of Ivanov~\cite{Iv07}, see also~\cite{Ra12},  we can avoid both $a_1$ and $a_2$
by coloring the plane with $9$ colors.
  The procedure retains the points in $X$ in a largest color class, namely at least $n/9$ points. 
(ii)~$a_2 > \sqrt3$. By a result of Csizmadia~\cite[Thm.~4]{Cs99}, in this case, 
  the combined multiplicity of the two smallest distances among the points is at most $4n$.
  See also~\cite[Ch.~5]{BMP05}.
  Let $\overline{d}$ denote the average degree of the corresponding graph.
  Then $\overline{d} \leq 8$ and by applying~\eqref{eq:caro-wei-avg} in $G(X,\{a_1,a_2\})$,
  we again obtain $g(n,\{a_1,a_2\}) \geq n/9$.
 
\medskip
\emph{Upper bound.}  For simplicity assume that $n$ is divisible  by $5$.
Consider $n/5$ widely spaced copies of a regular pentagon. Note that $a_1,a_2 \in A(X)$ are the two 
distances in the pentagon. At most one point can be retained from each pentagon in a subset that
avoids $a_1,a_2$.  This proves the upper bound.
\end{proof}

\begin{theorem} \label{thm:largest2}
  For the two largest distances $a_{s-1},a_s \in A(X)$ we have
  $g(n,\{a_{s-1},a_s\}) \geq \left \lfloor n/4 \right \rfloor$.
  On the other hand, there exists $X$, $|X|=n$, such that
  $g(X,\{a_{s-1},a_s\}) \leq \lfloor n/2 \rfloor$. 
\end{theorem}
\begin{proof}
\emph{Lower bound.} 
By Theorem~\ref{thm:max}, we can choose $\left\lfloor \frac{n}{2} \right\rfloor$
  to eliminate all diameters. The diameter of this point set is at most $a_{s-1}$, and thus,
  we can apply Theorem~\ref{thm:max} again.
Then an elementary computation yields that
\[ g(n,\{a_{s-1},a_s\}) \geq \left\lfloor \frac{\lfloor n/2 \rfloor }{2} \right\rfloor
= \left \lfloor \frac{n}{4} \right \rfloor.
\]

\medskip
\emph{Upper bound.} Consider a regular polygon  $X$ with $n$ vertices. By Observation~\ref{obs:monotone}
and Theorem~\ref{thm:max}, we have $g(X,\{a_{s-1},a_s\}) \leq g(X,\{a_s\}) \leq \lfloor n/2 \rfloor$. 
\end{proof}

\subsection{When the forbidden list contains extremal distances}

Previous lower bounds improve under this assumption:

\begin{proposition} \label{prop:extreme}
Let $D$ consist of $k$ distinct distances determined by an $n$-point planar set $X$.
\begin{enumerate} [{\rm (i)}]
\item If $a_1(X)\in D$, then
\[ g(X,D)\ge\frac{8}{31}\rhoC^{\,k-1} \cdot n. \]
In particular, when $k=2$ the coefficient is
$\frac8{31}\rhoC=0.0591\ldots$.
\item If $a_s(X)\in D$, then
\[ g(X,D)\ge\rhoC^{\,k-1}\left\lfloor\frac n2\right\rfloor. \]
\item If both $a_1(X),a_s(X)\in D$ (thus $k\ge2$), then
\[ g(X,D)\ge\rhoC^{\,k-2} \cdot \max\left\{1,\frac{8n-20}{47}\right\}. \]
\end{enumerate}
\end{proposition}
\begin{proof}
For (i), use the $8n/31$ lower bound in Theorem~\ref{thm:min} for the minimum
distance graph, and then successively avoid the other $k-1$ distances with Croft's
single-distance selection bound.

For (ii), first retain $\lfloor n/2\rfloor$ points avoiding the original diameter, and
then remove the remaining $k-1$ distances.

For (iii), use the lower bound $(8n-20)/47$ in Theorem~\ref{thm:min+max} to eliminate both
extremal distances simultaneously, followed by $k-2$ applications of the Croft bound.
No previously avoided distance can reappear when taking subsets.
\end{proof}

\subsection{Approximation algorithms} \label{subsec:approx2}

In this subsection we prove:

\begin{theorem} \label{thm:min+max-approx}
  Given a set $X$ of $n$ points in the plane: 
For any $\eps>0$, a $(0.5 -\eps)$-approximation for \textsc{MIS} in $G(X,\{a_1,a_s\})$ can be found
in $2^{O(1/\eps)} n^{O(1)}$ time.
\end{theorem}
\begin{proof}
  Let $\opt$ and $\opt_1$ denote the sizes of optimal solutions (for \textsc{MIS})
  in $G(X,\{a_1,a_s\})$ and $G(X,\{a_1\})$, respectively. Clearly, we have $\opt \leq \opt_1$. 
  Recall that by Theorem~\ref{thm:min} a $(1-\eps)$-approximation for \textsc{MIS} in $G(X,\{a_1\})$
  can be found in $2^{O(1/\eps)} n^{O(1)}$ time.

  In the first step, compute such a solution, denoted $X_1$; then $|X_1| \geq (1-\eps) \opt_1$.
  In the second step avoid $\diam(X)=a_s$ in $X_1$, if it exists, by removing at most half of the points
  in $X_1$. We thereby obtain a solution $X_2$ (for \textsc{MIS}) in $G(X,\{a_1,a_s\})$, whose size is
  \[ |X_2| \geq \frac12 |X_1| \geq \frac12 (1-\eps) \opt_1 \geq  \frac12 (1-\eps) \opt. \]
  Redefining $\eps$ yields the claimed approximation ratio.
\end{proof}

It is worth noting that, while we have a PTAS for each of the two problems, \ie,
one for \textsc{MIS} in  $G(X,\{a_1\})$, and one for \textsc{MIS} in  $G(X,\{a_s\})$,
composing the two does not necessarily yield a PTAS for \textsc{MIS} in  $G(X,\{a_1,a_s\})$.

\section{Interval restrictions} \label{sec:interval}

In Sections~\ref{sec:discrete-single} and~\ref{sec:discrete-multiple}
we stayed under the scenario that $D \subseteq A(X)$ is finite.
One can consider distance restrictions in the form of a union of  (forbidden) intervals.
Specifically, let $g(X,\J)$ denote the maximum size of a subset $Y \subseteq X$ avoiding every distance
contained in the union of a family $\J$ of intervals. The next example shows that one can avoid all distances
(\ie, all point-pairs) in a planar set whose convex hull contains a sufficiently large disk, \eg, of radius $1/4$,
by using only two arbitrarily short forbidden intervals.

\begin{observation}\label{obs:interval}
  For every $\eps>0$, there exists $X$ such that $\conv(X)$ contains a disk of radius $1/4$,
  and $g(X,\{[0,\eps], [1,1+\eps]\}) =1$. 
\end{observation}
\begin{proof}
Take a set of points sufficiently close to the three vertices of a unit regular triangle as $X$.
\end{proof}

Interestingly enough, we can give estimates that are much closer to each other for the problem
of reducing the diameter of a point set below a certain threshold. This problem is closely related
to a generalization of Borsuk's problem. To state our result, for any integer $k \geq 2$,
we denote by $d_k$ the smallest positive number such that any set $X \subset \RR^2$
of unit diameter can be decomposed into $k$ subsets of diameters at most $d_k$.

\begin{theorem}  \label{thm:half}
For all integers $k, n \geq 2$, we have $g(n,(d_k,1]) \geq n/k$.
In particular, if $X$ is a set of $n$ points of unit diameter in the plane,
then $g(X,(1/2,1]) \geq n/7$. On the other hand, there exists $X$, $|X|=n$, such that
  $G(X, (1/2,1]) \leq \lceil n/6 \rceil$. 
\end{theorem}
\begin{proof}
\emph{Lower bound.}
By definition, any set of unit diameter can be covered by $k$ sets of diameters at most $d_k$.
Thus, by the Pigeonhole Principle, at least one of these sets contains at least $\frac{n}{k}$ points.
In particular, by a result of Lenz~\cite{Le56a,Le56b}, any set of unit diameter in the plane can be covered
by $7$ sets of diameter $0.5$; see also~\cite{Fi10,TPV22}. Applying this result to $\conv(X)$, one of the $7$ parts
contains at least $n/7$ points in $X$.

\medskip
\emph{Upper bound.} 
Consider a regular pentagon of unit diameter and its center.
Note that the minimum distance between these points is $1/(2 \sin 72^\circ) = 0.526 \ldots$. 
If $n$ is divisible by $6$, place $n/6$ points in a small disk of radius $0.001$ centered at each of these six points.
The minimum distance between points in different disks is at least $0.52$ and the upper bound $n/6$ follows.
For arbitrary $n$ distribute the points among the same six clusters as evenly as possible.  The resulting upper
bound is $\lceil n/6 \rceil$. 
\end{proof}

A recent paper of Tolmachev, Protasov and Voronov~\cite{TPV22} collects the known estimates
about the values of $d_k$. According to this, $d_2=1$, $d_3=0.8660\ldots$,
  $d_4 = 0.7071\ldots$, $d_5 \leq 0.5953\ldots$, $d_6 \leq 0.5343\ldots$.
Observe that $d_k$ in the formulation of Theorem~\ref{thm:half} can be replaced by any $d$,
with $d_k \leq d \leq 1$.

\section{Concluding remarks} \label{sec:remarks}

\textbf{1.} 
The techniques used to handle the planar case immediately extend to $\RR^d$ for any fixed dimension $d$
using bounds on the chromatic number of $\RR^d$. 
The current best upper and lower bounds are due to  Larman and Rogers~\cite{LR72} and
Raigorodskii~\cite{Ra12}, respectively:
\begin{equation} \label{eq:chi}
  (1.239\ldots + o(1))^d \leq \chi(\RR^d) \leq (3 + o(1))^d.
\end{equation}

The proof of Theorem~\ref{thm:any-d} is analogous to that of Theorem~\ref{thm:unit} and
is therefore omitted. The resulting upper and lower bounds are exponentially  far from  each other.

\begin{theorem} \label{thm:any-d}
  For any finite set $X \subset \RR^d$, with $|X| = n$, and distance $x \in A(X)$
  we have $g(X,\{x\}) \geq c_d \, n$. 
  On the other hand, there exist $X \subset \RR^d$, with $|X| = n$,
  and $x \in A(X)$ such that $g(X,\{x\}) \leq c'_d \, n$.
  Here $c_d$ and $c'_d$ are constants that depend only on $d$.
\end{theorem}

\textbf{2.} 
Many interesting questions in relation to closing the remaining gaps remain.
Here we restrict ourselves to the plane:

\begin{enumerate} \itemsep 3pt
  
\item Theorem~\ref{thm:unit-approx} trivially yields a $0.2293$-ratio approximation algorithm
  (since the size of an optimal solution to {\textsc MIS} satisfies $\opt =\alpha(X) \leq n$). 
  It is an interesting question whether a substantially better approximation, perhaps closer to $1/2$, is possible.

\item For a finite unit distance graph $G$ in the plane, how small can the ratio $\alpha(G)/n$ be?
Can it approach  $\rho_{\rm \,Croft}$?
The first question was asked by Erd\H{o}s, see also~\cite{MRVZ25,Sz02}.
  
\item  Find an improved dependence of  $g(n,k)$ in terms of $k$; note that Theorem~\ref{thm:k}
  leaves an exponential gap.
  
\end{enumerate}

\medskip
It is perfectly possible that most of the conjectures listed below might be easier to refute
rather than confirm.
According to Conjectures~\ref{conj:min+max} and~\ref{conj:largest2} below, the values of $g(\cdot)$ for some
two distance restrictions are determined by suitable single distance restrictions.  

\begin{conjecture} \label{conj:min+max}
 $g(n,\{a_1,a_s\}) \approx g(n, \{a_1\}) = 5n/16 -O(1)$.
\end{conjecture}

\begin{conjecture} \label{conj:smallest2}
 $g(n,\{a_1,a_2\}) =n/5 -O(1)$.
\end{conjecture}

\begin{conjecture} \label{conj:largest2}
 $g(n,\{a_{s-1},a_s\}) \approx g(n,\{a_s\}) = n/2 -O(1)$.
\end{conjecture}

\begin{conjecture} \label{conj:half}
$ g(n,(1/2,1]) = n/6 - O(1)$. 
\end{conjecture}

\subparagraph{AI Statement.}
The results in this paper have been developed by the authors.
LLMs have been used only for final checks and small adjustments after the work was completed.

\end{document}